\def\be{\begin{equation}}
	\def\ee{\end{equation}}
\def\ba{\begin{array}}
	\def\ea{\end{array}}

\documentclass[prl,showpacs,twocolumn,amsmath]{revtex4}
\usepackage{amsmath}
\usepackage{amsfonts}
\usepackage{mathrsfs}
\usepackage{amssymb}
\usepackage{pifont}
\usepackage{epsfig,subfigure,dsfont,amsthm,amsbsy,mathrsfs,amscd}
\usepackage{epstopdf}
\usepackage{bbm}
\usepackage{braket}
\usepackage{caption}
\graphicspath{{photography/}}{\large }
\usepackage{graphicx}
\usepackage{pgfplots}
\usepackage{booktabs}
\usepackage{makecell}
\def\qed{\leavevmode\unskip\penalty9999 \hbox{}\nobreak\hfill
	\quad\hbox{\leavevmode  \hbox to.77778em{%
			\hfil\vrule   \vbox to.675em%
			{\hrule width.6em\vfil\hrule}\vrule\hfil}}
	\par\vskip3pt}
\usepackage{leftidx}
\newtheorem{thm}{\bf Theorem}

\input amssym.def
\begin{document}
	\title{\large\bf A note on  entanglement detection via the generalized realignment moments} 
	\author{Xiaofen Huang$^{1, \dag}$, Xishun Zhu$^{1}$,  Bin Chen$^{2}$, Naihuan Jing$^{3}$ and  Shao-Ming Fei$^{4}$ }
	\affiliation{ ${1}$ School of Mathematics and Statistics, Hainan Normal University, Haikou, 571158, China \\
		$2$ School of Mathematical Sciences, Tianjin Normal University, Tianjin 300387, China \\
        $3$ Department of Mathematics, North Carolina State University, Raleigh, NC 27695, USA \\
        $4$ School of Mathematics Sciences, Capital Normal University, Beijing100048, China \\
		$^{\dag}$ Correspondence to huangxf1206@163.com}
	
	\bigskip
	\bigskip

\begin{abstract}
The experimental detection of quantum entanglement is of great importance in quantum information processing. We present two separability criteria based on the generalized realignment moments. By incorporating additional parameters, these criteria prove to be more flexible and stronger than some of existing ones. Detailed examples are given to demonstrate their availability and feasibility for entanglement detection.
\end{abstract}

\maketitle
	
\section{I. Introduction}
Quantum entanglement raised the cental problem in the EPR argument \cite{EPR} to demonstrate the incompleteness of quantum mechanics. As an essential resource in quantum information science, it enables applications including quantum algorithms \cite{chuang, algo}, quantum cryptography \cite{cryp1, cryp2}, quantum simulation \cite{simula} and quantum teleportation \cite{telep}. A bipartite quantum state $\rho$  associated with the Hilbert space $\mathcal{H}_A\otimes \mathcal{H}_B$ is said to be separable if it has a convex representation of product states,
\begin{equation}
\rho=\sum_{i}p_i\rho_{i}^{A}\otimes \rho_i^{B},
\end{equation}
where $p_i$ is a probability distribution,  $\rho_{i}^{A}$  and $ \rho_i^{B}$ are the density matrices of the subsystems $\mathcal{H}_A$ and $\mathcal{H}_B$, respectively. Otherwise $\rho$ is entangled.

Detecting whether a quantum state is entangled or not is a fundamental problem in the theory of quantum entanglement. The celebrated PPT (positive partial transpose) criterion says that the partially transposed matrix of any separable state must be semipositive \cite{PPT}. The PPT criterion is both necessary and sufficient for $2\otimes 2$ and $2\otimes 3$ quantum systems \cite{iso}, but not sufficient for high-dimensional states, as the bound entangled states are PPT states. The realignment criteria \cite{realig} detects the entanglement of many bound entangled states. In Refs.\cite{rea1, rea2} the authors improved the existing realignment criterion by analyzing the symmetric functions of the singular values of the realigned matrices. The authors in Ref. \cite{rea3} established a few entanglement criterion for two-qubit and two-qudit systems based on the realignment operation. In Ref. \cite{rea4}, realignment separability criterion assisted with filtration was proposed for detecting the continuous-variable entanglement.

Recently some further elegant results on separability criteria have been derived \cite{nega, cova, guhne, prl125, prl127, zhang2008}. Elben \textit{et al.}\cite{prl125} proposed a separability criterion based on moments given by the partially transposed matrix. Then more separability criteria based on the Hankel matrices involving all the PT-moments \cite{prl127},  the realignment moments \cite{zhang2022,agg2024,wang2024} and the moments with respect to positive maps \cite{wang2022, zhang2025}. In Ref. \cite{adh2024} the authors introduced an entanglement criterion by estimating the bound of realignment moments, which detects both non-PPT entangled states and bound entangled states. The authors in Ref. \cite{zhao2025} proposed entanglement criteria for multipartite entanglement based on realignment moments. In Ref. \cite{moment1}, an efficient detection of genuine multipartite entanglement has been presented by using the moments of positive maps.

In this paper, we proposed two entanglement criteria based on the moments of the generalized realignment matrix with parameters. Detailed examples are given to illustrate the power, effectiveness and advantages of our approach in detecting bipartite  entanglement.
	
\section{II. Generalized realignment criterion with parameters based on moments}
We first review some concepts about the realignment. Let $\mathbb{ C}^{m\times n}$ denote the set of all $m\times n$ matrices over complex field $\mathbb{ C}$, and $\mathbb{ R}$ the field of real numbers.
For a matrix $A=(a_{i, j})\in \mathbb{ C}^{m\times n}$, the vectorization of matrix $A$ is defined as $\mathrm{Vec}(A)=(a_{11}, \dots, a_{m1}, a_{12}, \dots, a_{m2}, \dots, a_{1n}, \dots, a_{mn})^T$, where $T$ stands for the transpose. Let $Z$ be an $m\times m$ block matrix with sub-blocks $Z_{i, j}\in\mathbb{ C}^{n\times n}$, $i, j=1, 2, \dots, m$. The realigned matrix $\mathcal{R}(Z)$ of $Z$ is defined by
$$
\mathcal{R}(Z)=\begin{pmatrix}
\mathrm{Vec}(Z_{1, 1})^T\\
\vdots\\
\mathrm{Vec}(Z_{m, 1})^T\\
\vdots\\
\mathrm{Vec}(Z_{1, m})^T\\
\vdots\\
\mathrm{Vec}(Z_{m, m})^T\\
\end{pmatrix}.
$$

The well-known realignment criterion \cite{realig} says that any separable state $\rho$ in $\mathbb{C}^{d_A}\otimes \mathbb{C}^{d_B}$ satisfies $\|\mathcal{R}(\rho)\|_{\mathrm{tr}}\leq 1$, where $\|A\|_{\mathrm{tr}}=\mathrm{tr}(\sqrt{A^{\dagger}A})$ is the trace norm of $A$. In  \cite {zhang2008} Zhang \textit{et al.} showed that for any separable state $\rho_{AB}$, the following inequality holds,
\begin{equation}
\|\mathcal{R}(\rho_{AB}-\rho_A\otimes \rho_B) \|_{\mathrm{tr}}\leq \sqrt {1-\mathrm{tr} \rho_A^2}\sqrt {1\mathrm{-tr} \rho_B^2}.
\end{equation}

Furthermore, Shen \textit{et al.} \cite{shen2015} considered the block matrix,
\begin{equation}
	\mathcal{N}_{\alpha, l}^G(\rho)=
	\begin{pmatrix} G-\alpha^2E_{l\times l} & \alpha_{\omega_l}(\rho_B)^T\\
		\alpha_{\omega_l}(\rho_A) & \mathcal{R}(\rho)
	  \end{pmatrix},
\end{equation}
where $E_{l\times l}$ is an $l\times l$ matrix with all the entries being $1$, $G$ is a matrix such that $G-\alpha^2E_{l\times l}$ is positive semidefinte, $\beta\in \mathbb{R}$ and $\omega_l(X)$ is defined by
\begin{equation}
	\omega_l(X)=\underbrace{(\textrm{Vec}(X), \cdots, \textrm{Vec}(X) )}_{l \text{ columns}}.
\end{equation}
It is shown that if $\rho$ is a separable, then
\begin{equation}
\|\mathcal{N}_{\alpha, l}^G(\rho) \|_\mathrm{{tr}}\leq 1+\mathrm{tr}(G).	
\end{equation}

Recently, Shi \textit{et al.} \cite{shi2023} constructed the following matrix,
\begin{equation}
\mathcal{M}_{\alpha, \beta}(\rho)=\begin{pmatrix}\alpha \beta  & \alpha \textrm{Vec}(\rho_B)^T\\
	\beta \textrm{Vec}(\rho_A) & \mathcal{R}(\rho)
\end{pmatrix},
\end{equation}
where $\alpha, \beta \in\mathbb{R}$, $\rho_A$ and $\rho_B$ are the reduced density matrices of the $A$ and $B$ subsystems, respectively. They showed that if $\rho$ is separable, the following inequality holds,
\begin{equation}
\|\mathcal{M}_{\alpha, \beta}(\rho)\|_\mathrm{{tr}}\leq \sqrt{(\alpha^2+1)(\beta^2+1)}.	
\end{equation}
Sun \textit{et al.} \cite{sun2024} generalized this criterion 
by constructing the following block matrix $\mathcal{M}^l_{\alpha, \beta}(\rho)$,
\begin{equation}
	\mathcal{M}^l_{\alpha, \beta}(\rho)=\begin{pmatrix}\alpha \beta E_{l\times l}  & \alpha \omega_l(\rho_B)^T\\
		\beta\omega_l(\rho_A) & \mathcal{R}(\rho)
	\end{pmatrix},
\end{equation}
where $\alpha$ and $\beta$ are arbitrary real numbers, $l$ is a natural number. They showed that if $\rho$ is separable, then 
\begin{equation}
	\|\mathcal{M}^l_{\alpha, \beta}(\rho)\|_\mathrm{{tr}}\leq \sqrt{(l\alpha^2+1)(l\beta^2+1)}.	
\end{equation}

We present separability criteria based on the moments of the generalized realignment matrix $\mathcal{M}^l_{\alpha, \beta}(\rho)$ by using the H\"older inequality \cite{Hold} and Schatten-$p$ norm \cite{schatten} of matrices. We define the moments of generalized realignment matrix,
\begin{equation}
a_k: ={\rm{tr}}[\mathcal{M}^l_{\alpha, \beta}(\rho)\mathcal{M}^l_{\alpha, \beta}(\rho)^{\dag}]^{\frac{k}{2}},
\end{equation}
where $k=0, 1, \cdots, (d_A^2-1)(d_B^2-1)$.
For convenience we set $a_0=(d_A^2-1)(d_B^2-1)$ and denote $\textbf{a}=(a_0, a_1, \cdots, a_{(d_A^2-1)(d_B^2-1)})$.
We have the following conclusion.

\begin{thm}\label{thm1}
 For a separable bipartite quantum state $\rho$, the following inequality holds,
\begin{equation}\label{bi1}
a_2^2\leq \sqrt{(l\alpha^2+1)(l\beta^2+1)}a_3,
\end{equation}
where $\alpha$ and $\beta$ are arbitrary real numbers, $l$ is a natural number.
\end{thm}

\textbf{Proof}
For $n$-dimensional vectors $u=(u_i)$ and $v=(v_i)$, the H$\rm{\ddot{o}}$lder inequality says that
\begin{equation}
|\langle u, v\rangle|=\sum_{i=1}^n u_iv_i\leq \| u\|_{l_p}\|v\|_{l_q},
\end{equation}
where $p, q\geq 1$,  $\frac{1}{p}+\frac{1}{q}=1$,  and $\|u\|_{l_p}:=(\sum_i|u_i|^p)^{\frac{1}{p}}$ is the $l_p$ norm of $u$.
The Schatten-$p$ norm of a Hermitian matrix $M$ is given by
\begin{equation}
\|M\|_p=\Big({\rm{tr}}(MM^{\dag})^{\frac{p}{2}}\Big)^{\frac{1}{p}}=(\sum_i s_i^p)^{\frac{1}{p}},
\end{equation}
where $s_i$'s are the singular values of the matrix $M$. Note that $\|M\|_1$ reduces to the trace norm $\|M\|_{\textrm{tr}}$ of $M$.

By definition of the Schatten-$p$ norm, one has that $\|M\|_2^4=(\sum_{i}s_i^2)^2$, $\|M\|_1=\sum_is_i$, and $\|M\|_3^3=\sum_is_i^3$. Let $u_i=s_i^{\frac{1}{2}}$ and $v_i=s_i^{\frac{3}{2}}$ in the H\"older inequality, then
\begin{equation}
\sum_{i}s_i^2=\sum_i s_i^{\frac{1}{2}}s_i^{\frac{3}{2}}\leq (\sum_is_i)^{\frac{1}{2}}(\sum_is_i^3)^{\frac{1}{2}},
\end{equation}
i. e.,
\begin{equation}
(\sum_{i}s_i^2)^2\leq (\sum_is_i)(\sum_is_i^3).
\end{equation}
Therefore, for any Hermitian matrix $M$ one obtains
\begin{equation}\label{bi3}
\|M\|_2^4\leq \|M\|_1 \|M\|_3^3.
\end{equation}

Take $M=[\mathcal{M}^l_{\alpha, \beta}(\rho)\mathcal{M}^l_{\alpha, \beta}(\rho)^{\dag}]^{\frac{1}{2}}$. We have
$\|M\|_2^4=a_2^2$, $\|M\|_1=a_1$ and $\|M\|_3^3=a_3$. By inequality (\ref{bi3}), for any bipartite state $\rho$ we have
\begin{equation}\label{bi33}
a_2^2\leq a_1a_3.
\end{equation}
If $\rho$ is separable, the generalized realignment matrix  satisfies $\|\mathcal{M}^l_{\alpha, \beta}(\rho)\|_{\textrm{tr}}={\rm {tr}} (\mathcal{M}^l_{\alpha, \beta}(\rho)\mathcal{M}^l_{\alpha, \beta}(\rho)^{\dag})^\frac{1}{2}=a_1\leq \sqrt{(l\alpha^2+1)(l\beta^2+1)}$ \cite{sun2024}. Substituting it into the inequality (\ref{bi33}), one obtains the inequality (\ref{bi1}).
 $\Box$

The separability criterion in Theorem \ref{thm1} involves the first three moments. To adopt more moments
to detect entanglement, we construct two families of semipositive matrices named the Hankel matrices \cite{prl127}, 
1).  $H_k(\textbf{a})$  with entries $[H_k(\textbf{a})]_{ij}=a_{i+j}$, where $k=1, 2, \cdots, \lfloor \frac{d_Ad_B}{2}\rfloor$ and $i, j=0, 1, \cdots, k$; 2). $B_l(\textbf{a})$  with entries  $[B_l(\textbf{a})]_{mn}=a_{m+n+1}$ for $l=1, 2, \cdots, \lfloor \frac{d_Ad_B-1}{2}\rfloor$ and $m, n=0, 1, 2, \cdots, l$, where $\lfloor \cdot \rfloor$ stands for the integer function. We have the following separability criteria.

\begin{thm}\label{thm2} For any separable bipartite state $\rho$, the Hankel matrices satisfy
the following inequalities: 
\begin{equation}\label{bi4}
\widehat{H_k}(\textbf{a})\geq 0,~~ \widehat{B_l}(\textbf{a})\geq 0,
\end{equation}
where $\widehat{H_k}(\textbf{a})$ and $\widehat{B_l}(\textbf{a})$  are the matrices obtained by replacing $a_1$ with $\sqrt{(l\alpha^2+1)(l\beta^2+1)}$  in the Hankel matrices  $H_k(\textbf{a})$  and $B_l(\textbf{a})$.
\end{thm}

\textbf{Proof}  Set $H=[\mathcal{M}^l_{\alpha, \beta}(\rho)\mathcal{M}^l_{\alpha, \beta}(\rho)^{\dag}]^{\frac{1}{2}}$. We introduce two matrix vectors
\begin{eqnarray*}
  \textbf{x}: &=& (x_0, x_1, \cdots, x_{\lfloor \frac{d_Ad_B}{2}\rfloor}) \\
   &=& (I=H^0, H, H^2,  \cdots, H^{\lfloor \frac{d_Ad_B}{2}\rfloor}),
\end{eqnarray*}
and
\begin{eqnarray*}
  \textbf{y}: &=& (y_0, y_1, \cdots, y_{\lfloor \frac{d_Ad_B-1}{2}\rfloor}) \\
   &=& (H^{\frac{1}{2}}, H^{\frac{3}{2}}, \cdots, H^{\lfloor \frac{d_Ad_B-1}{2}\rfloor}).
\end{eqnarray*}
 Using the Hilbert-Schmidt inner product of matrices, we construct two Gram matrices in terms of $\textbf{x}$ and $\textbf{y}$,
$\langle x_i, x_j\rangle={\rm{tr}} H^{i+j}=a_{i+j}$ and $ \langle y_i, y_j\rangle={\rm{tr}} H^{i+\frac{1}{2}}H^{j+\frac{1}{2}}={\rm{tr}} H^{i+j+1}=a_{i+j+1}$, which are just the Hankel matrices $H_k(\textbf{a})$ and $B_l(\textbf{a})$, respectively. As Gram matrices are always positive semidefinite, we have $H_k(\textbf{a})\geq 0$ and $B_l(\textbf{a})\geq 0$ for any state.

Since the first moment $a_1=\|\mathcal{M}^l_{\alpha, \beta}(\rho)\|_{\textrm{tr}} \leq \sqrt{(l\alpha^2+1)(l\beta^2+1)}$ for any separable bipartite $\rho$ \cite{sun2024}, by replacing all $a_1s$ in $H_k(\textbf{a})$ and  $B_l(\textbf{a})$ with $\sqrt{(l\alpha^2+1)(l\beta^2+1)}$, we obtain the inequalities (\ref{bi4}).
 $\Box$

If we choose $a_1=\sqrt{(l\alpha^2+1)(l\beta^2+1)}$, by Theorem \ref{thm2} the Hankel matrices $\widehat{B_1}(\textbf{a})=\begin{pmatrix} a_1 & a_2\\ a_2 & a_3 \end{pmatrix}\geq 0$. Theorem 2 reduces to the inequality given in Theorem \ref{thm1}.

The separable criteria given by Theorem 1 and Theorem 2 can be directly generalized to the case of multipartite states.
It is worth noting that our separability criterion given in Theorem 1 reduces to the criterion proposed by Zhang\textit{ et al.} \cite{zhang2022} in the case of $l=0$. The following examples show the power of our separability criteria in detecting entanglement.

\textbf{Example 1 }
Let us consider the following  bipartite state,
\[
\rho=p|\varphi\rangle\langle \varphi|+\frac{1-p}{4}I,~~ p\in[0,1],
\]
where $|\varphi\rangle=\frac{1}{\sqrt{2}}(|00\rangle+|11\rangle)$ is the maximally entangled state, and $I$ is the $4\times 4$ identity matrix.  Set $f=a_2^2- \sqrt{(l\alpha^2+1)(l\beta^2+1)}a_3$ to be difference between the right and left sides of the inequality in Theorem 1. Fig.1 shows that function $f$  increases with the state parameter $p$, and smaller values of $\alpha$ yield larger $f$. Setting $l=1$ and $\alpha=\beta=\frac{1}{729}$, we obtain by straightforward calculation that
the state $\rho$ is entangled for $p\in [0.44, 1]$.
\begin{figure}[ht]
    \centering
    \includegraphics[width=0.35\textwidth]{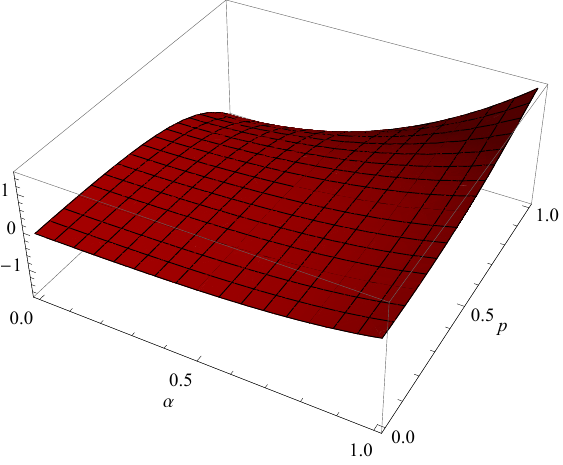}
    \caption{$f=a_2^2- \sqrt{(l\alpha^2+1)(l\beta^2+1)}a_3$ as a function of $p$ and $\alpha$
    with  $l=1$ and $\beta=\alpha$. }
    \label{fig:example}
\end{figure}

The realignment moments based separability criterion  given in \cite{wang2024} says that the state is entangled  within the range of $p\in (0.54, 1]$. While the witnesses based separability criterion given in Ref.\cite{WYD2025} gives rise to that the state is  entangled when $p\in (\sqrt{\frac{1}{3}}, 1]$. It is clear that our criterion given in Theorem 1 detects better the entanglement of this state.  According to the realignment criterion, $\rho$ is entangled in the case $p \in (\frac{1}{3}, 1]$. This implies that, for the sake of experimental operability and convenience, experimentally-based entanglement criteria sacrifice some precision, thus exhibiting a weaker ability to detect entanglement compared to the original theoretically-founded realignment criterion.

\textbf{Example 2}
Consider the Werner states in $\mathcal{H}_d\otimes \mathcal{H}_d$,
\begin{equation}
\rho_W=\frac{1}{d^3-d}[(d-p)I\otimes I+(dp-1)F],
\end{equation}
where $-1\leq p\leq 1$ and $F$ is the flip operator given by $F(\phi\otimes \varphi)=\varphi\otimes \phi$.
These states are separable iff $p\geq 0$ \cite{werner}.

Set $d=2$. It is verified that the generalized realignment matrix $\mathcal{M}_{\alpha, \beta}(\rho)$ has four singular values, three of which are equal to $|2p - 1|/6$. As shown in Fig. \ref{fig:wer2},  smaller values of $\alpha$ and greater values of $l$ lead to stronger entanglement detection by Theorem 1.
The value of $f=a_2^2- \sqrt{(l\alpha^2+1)(l\beta^2+1)}a_3$  is greater than zero when $p \in [-1, -0.163744]$ if we choose $\alpha=1/729$ and $l=3$, indicating that our entanglement criterion can identify the entanglement of Werner state for $p \in [-1, -0.163744]$.
\begin{figure}[ht]
    \centering
    \includegraphics[width=0.45\textwidth]{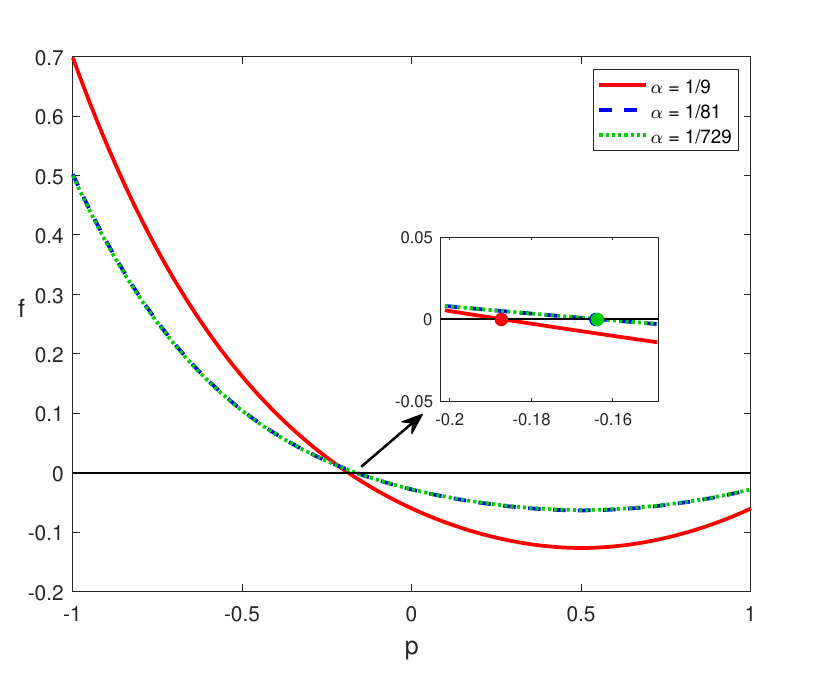}
    \caption{The entanglement of Werner state $\rho_W$ is fully detected by Theorem 1 when $l=3$ and $\beta=3\alpha$. }
    \label{fig:wer2}
\end{figure}

\textbf{Example 3}
Consider the two-qubit isotropic state given in \cite{zhao2010},
\begin{equation}
\rho_b=\frac{1-b}{3}I\otimes I+\frac{4b-1}{3}\ket{\varphi}\bra{\varphi}, 0\leq b\leq 1.
\end{equation}
If we set $l=1$, $\alpha=1/8$ and $\beta=1/16\sqrt{2}$, $\rho_b$ is entangled when $b>0.501550$ by the entanglement criterion proposed in Theorem 2, which is exactly the same result as the one from the realignment and PPT criterion directly, and stronger
than the result $b\geq 0.608594$ given in \cite{agg2024}.

\section{III.Conclusions}
We have investigated the entanglement detection of bipartite quantum states. Based on the moments of the realigned matrix of a density matrix, we have introduced the generalized realignment moments,  and proposed two feasible separability criteria for any dimension bipartite quantum states. The discriminant in these criteria can also be represented in terms of PT moments. Therefore, these criteria
can also be experimentally implemented. The advantage of our separability criteria lies in its parameterization of the realignment matrix, which leads to  more refined and flexible criteria.
At last, detailed examples show that our separability criteria are effective to detect entanglement.

\bigskip
\noindent{\bf Acknowledgments}

This work is supported by the National Natural
Science Foundation of China (NSFC) under Grant nos. (12204137, 12171303, 12171044); the China Scholarship Council (CSC);
the Natural Science Foundation of Hainan Province under Grant No. 125RC744; the Talent and Research Star-up Fund of Hainan Normal University No. HSZK-KYQD-202401; the specific research fund of the Innovation Platform for Academicians of Hainan Province.

\bigskip
\textbf{Date Availability Statement}
All data generated or analysed during this study are available from the corresponding author on reasonable request.

\bigskip
\textbf{Conflict of Interest Statement}
We declare that we have no conflict of interest.

\end{document}